\documentclass[preprint,12pt]{elsarticle}

\usepackage[T1]{fontenc}
\usepackage[utf8]{inputenc}

\usepackage{lmodern}
\usepackage{microtype}
\usepackage{xcolor}

\usepackage{amsmath,amssymb,amsfonts}
\usepackage{amsthm}
\usepackage{mathtools}

\usepackage{tikz}
\usepackage{tikz-cd}
\usetikzlibrary{arrows.meta, calc, positioning}
\tikzcdset{arrow style=tikz, diagrams={>={latex}}}
\usepackage{geometry}
\usepackage{graphicx}

\usepackage{titlesec}
\makeatletter
\let\paragraph\elsparagraph
\@ifundefined{elssubparagraph}{}{\let\subparagraph\elssubparagraph}
\makeatother

\usepackage{empheq}

\usepackage{bm}
\usepackage{enumitem}
\usepackage{hyperref}
\usepackage[nameinlink,noabbrev]{cleveref}

\theoremstyle{plain}
\newtheorem{theorem}{Theorem}[section]
\newtheorem{proposition}[theorem]{Proposition}

\newtheorem{corollary}[theorem]{Corollary}

\theoremstyle{definition}
\newtheorem{definition}[theorem]{Definition}
\newtheorem{remark}[theorem]{Remark}
\theoremstyle{note}

\newcommand{\R}{\mathbb{R}}
\newcommand{\dR}{\mathrm{dR}}

\newcommand{\dd}{\mathrm{d}}
\newcommand{\grad}{\nabla}

\newcommand{\wtd}{\tilde d}
\newcommand{\wtdel}{\tilde\delta}
\newcommand{\wtD}{\tilde\Delta}

\journal{}

\begin{document}

\begin{frontmatter}

\title{Homothetic Hodge–de Rham Theory and a Geometric Regularization of Elliptic Boundary Value Problems}

\author{Fereidoun Sabetghadam}
\address{%
  \textit{Mechanical Engineering Faculty, Science and Research Branch, IAU, Tehran, Iran}\\
  \textit{Email: fsabet@srbiau.ac.ir}
}

\begin{abstract}
We introduce a homothetic extension of classical Weyl–integrable geometry by generalizing the usual linear gauge transformations to affine homothetic transformations centered at a distinguished harmonic, scale–invariant form $\alpha_d$. After re–linearizing these affine gauge transformations via a suitable shift of variables, we obtain a twisted exterior calculus that is
structurally equivalent to the Witten deformation of the de~Rham complex. On this basis we develop a corresponding homothetic Hodge theory: we define a twisted adjoint and homothetic Laplacian, and prove a homothetic Hodge decomposition theorem on compact Riemannian manifolds. In the context of partial differential equations, we show that the scalar homothetic Laplacian provides a rigorous diffuse–interface (volume–penalization) representation of elliptic boundary value problems. Modeling the Weyl scale field as a fixed distribution localized near a hypersurface, the resulting
lower–order geometric terms form a penalty layer that enforces Dirichlet, Neumann, or Cauchy data within a single geometric equation. This formulation yields consistent weak solutions even in the presence of classically incompatible Cauchy data. As an application, we construct a non–singular model for point sources in elliptic field equations, which preserves the correct Coulombian far–field while removing the core singularity and yielding finite field energy.
\end{abstract}


\end{frontmatter}

\section{Introduction}\label{sec:intro}

Classical Weyl geometry was originally proposed as an elegant attempt to unify gravitation and electromagnetism by introducing local scale invariance \cite{Weyl1918}. While its original physical interpretation was largely superseded, Weyl-integrable geometry has since found profound applications in modern physics, particularly in scalar-tensor theories of gravity \cite{RomeroFonsecaNetoPucheu2012}, conformal manifolds, and symmetry reduction in gauge theories \cite{KarananasMonin2016, Sabetghadam2025}.

Beyond its role in scalar–tensor gravity and conformal geometry,
Weyl–integrable structures naturally accommodate generalized notions of scale symmetry
that are well suited to geometric formulations of differential operators and field
equations. Motivated by this perspective, we introduce in this work a \emph{homothetic} extension
of Weyl–integrable geometry. Instead of the usual linear scaling of a $p$-form, $\alpha \mapsto e^{w\sigma}\alpha$, we postulate an affine transformation characterized by a distinguished, harmonic, scale-invariant $p$-form, $\alpha_d$, which acts as a ``homothety center,'' i.e., a fixed point of the affine scaling. By re-linearizing this transformation via shifting the  variables, we naturally induce a twist in the exterior calculus. Interestingly, this twisted differential conceptually and structurally coincides with the exact Witten deformation of the de Rham complex \cite{Witten1982, Roe1998}. This geometric structure provides the mathematical foundation for our first main result: the construction of a homothetic Hodge--de Rham theory, complete with a corresponding twisted adjoint, a homothetic Laplacian, and a modified Hodge decomposition theorem on compact Riemannian manifolds.

Beyond its intrinsic geometric interest, this homothetic Hodge theory offers a novel and rigorous mechanism for tackling classical elliptic boundary value problems. Traditional approaches to imposing boundary conditions typically rely on strict domain truncation and trace operators. Here, we demonstrate that the scalar homothetic Laplace equation can act as a natural volume-penalization or diffuse-interface method \cite{Mittal2008}. By considering the Weyl potential $\lambda$ as a localized, fixed distribution concentrated along a hypersurface, the lower-order terms of the homothetic operator act as a thin penalty layer. We show that this single geometric equation can seamlessly enforce Cauchy, Dirichlet, or Neumann boundary data, yielding mathematically consistent weak transmission solutions even when Cauchy data are classically incompatible.

A compelling mathematical physics application is the regularization of singular point sources \cite{BornInfeld1934, Jackson1998}. By framing the point source within our homothetic Laplace equation and enforcing boundary data on a regularized hollow-sphere interface, we obtain a non-singular model that decouples the interior and exterior fields. The resulting potential maintains the correct Coulombic behavior in the far field while eliminating the singularity at the origin, yielding a finite total field energy.

The paper is organized as follows. Section \ref{sec:setting} reviews classical Weyl-integrable geometry and introduces our homothetic extension. In Section \ref{sec:relation-witten}, we define the twisted differential operators and establish their relation to Witten deformations. Section \ref{sec:homothetic-derham} develops the homothetic de Rham complex, followed by the formulation of the homothetic Hodge theory and its decomposition theorem in Section \ref{sec:homothetic-hodge}. In Section \ref{sec:scalar-pde}, we derive the explicit scalar homothetic Laplace equation. Section \ref{PDE-Appl} explores the PDE applications, demonstrating how localized scale fields enforce interface conditions and resolving the point-source singularity. Finally, Section \ref{sec:Conclusions} concludes the paper, and \ref{lambda_bcs} provides the ODE analysis governing the localization of the scale field.

\section{A homothetic extension of Weyl-Integrable Geometry}
\label{sec:setting}
The formulation is provided consistent with our previous formulations \cite{Sabetghadam2025}, which in turn, is consistent with the major scalar-tensor theories of gravity (see e.g., \cite{RomeroFonsecaNetoPucheu2012}). 

\subsection{Classical Weyl-integrable structure}

Let $M$ be a smooth $d$-dimensional manifold. A \emph{Weyl structure} \cite{Weyl1918} on $M$ consists of a conformal class of (pseudo)-Riemannian metrics $[g]$ together with a one-form $\phi\in\Omega^{1}(M)$, and a torsion-free connection $\hat\nabla$ satisfying the \emph{non-metricity} condition
\begin{equation}\label{eq:weyl-nonmetricity}
\hat\nabla g \;=\; -2\,\phi\otimes g .
\end{equation}
Given $(M,[g],\phi)$, a \emph{Weyl gauge} (scale) transformation is specified by 
\begin{equation}\label{eq:weyl-gauge}
g \longmapsto e^{-2\sigma}g,
\qquad
\phi \longmapsto \phi + d\sigma,
\end{equation}
in which $\sigma\in C^\infty(M)$. Under \eqref{eq:weyl-gauge}, the connection $\hat\nabla$ remains unchanged, and \eqref{eq:weyl-nonmetricity} is preserved.

The Weyl structure is called \emph{integrable} if $d\phi=0$ \cite{SalimSautu1996,RomeroFonsecaNetoPucheu2012}. Therefore, it is \emph{exact} globally if $\phi=d\lambda$ for some smooth function $\lambda\in C^\infty(M)$. 

Along with the gauge freedom \eqref{eq:weyl-gauge}, one usually consider the transformation
\begin{equation}\label{eq:Transf_Classic}
\alpha \mapsto e^{-w\sigma}\alpha,
\end{equation}
for a differential $p$-form $\alpha \in \Omega^p(M)$ \cite{KarananasMonin2016}, in which $w$ is called the \emph{Weyl weight} and in physics it usually determines such that a physical theory be scale invariant.   

In the globally exact case $\phi=d\lambda$, choosing $\sigma=-\lambda$ in \eqref{eq:weyl-gauge} produces the \emph{Einstein (Riemann) gauge}, in which
\begin{equation}\label{eq:einstein-gauge}
\phi_{E}=0,
\qquad
g_{E}=e^{2\lambda}g .
\end{equation}
In this gauge, the Weyl connection $\hat\nabla$ reduces to the Levi--Civita connection of the metric $g_{E}$.

\subsection{Homothetic extension: affine scaling with a fixed point}
Here we propose our extension to the above classical Weyl integrable geometry by generalization of the linear transformation \eqref{eq:Transf_Classic} to a homothetic one.

Fix $p\in\{0,1,\dots,d\}$. Assume there exists a distinguished form $\alpha_d\in\Omega^p(M)$ that is \emph{scale invariant} under \eqref{eq:weyl-gauge}. We interpret $\alpha_d$ as a ``homothety center'' (fixed point). Let $\alpha_E\in\Omega^p(M)$ denote a $p$-form field in the Einstein gauge (i.e., for $\phi=\phi_E=0$). Given a Weyl weight $w\in\R$, first we define the \emph{homothetically scaled form}
\begin{equation}\label{eq:affine-homothety}
\alpha \;:=\; e^{-w\lambda}\alpha_E + (1-e^{-w\lambda})\alpha_d
\;=\; e^{-w\lambda}(\alpha_E-\alpha_d)+\alpha_d.
\end{equation}
Then, we postulate the gauge transformation
\begin{equation}\label{eq:Weyl-extension}
\boxed{
g_E \mapsto g:=e^{-2\lambda}g_E, \qquad \phi=0 \mapsto 0+d \lambda, \qquad \alpha_E \mapsto \alpha.}
\end{equation}
Therefore, for the Weyl-integrable manifold $(M,[g],\phi)$, the Weyl gauge (scale) freedom \eqref{eq:weyl-gauge}--\eqref{eq:Transf_Classic} is changed to 
\begin{equation}\label{eq:Weyl-extension-2}
g \mapsto g:=e^{-2\sigma}g, \quad \phi \mapsto \phi+d \sigma, \quad \alpha_E \mapsto e^{-w(\lambda+\sigma)}\alpha_E+(1-e^{-w(\lambda+\sigma)})\alpha_d.
\end{equation}

We call this setting the homothetic extension of the Weyl integrable geometry, and we shall show that a consistent Hodge-de Rham theory can be constructed on it.
\begin{remark} ~

\begin{enumerate}
\item For the special case $\alpha_d=0$, the conventional Weyl gauge freedom \eqref{eq:weyl-gauge}-\eqref{eq:Transf_Classic} is retrieved for $(M,[g],\phi)$
\begin{equation}
g \mapsto g:=e^{-2\sigma}g, \qquad \phi \mapsto \phi+d \sigma, \qquad \alpha \mapsto e^{-w\sigma}\alpha.
\end{equation}
However, for general $\alpha_d$, \eqref{eq:Weyl-extension} introduces a homothety on the affine space $\alpha_d+\Omega^p(M)$, with the homothety center $\alpha_d$ and scale factor $e^{-w\lambda}$.
\item On the above setting, the metric tensor is a 2-form with zero homothety center $g_d=0$, and the Weyl weight $w=2$.
\item We should emphasize that $\hat{\nabla}$ is still invariant under the scale freedom \eqref{eq:Weyl-extension-2}.  
\end{enumerate}

\end{remark}

However, since transformation \eqref{eq:affine-homothety} is not linear, but affine, the transformation of the differential forms in \eqref{eq:Weyl-extension-2} is not linear as well. So,  we first re-linearize it.

\subsection{Linearization by shifting}
To recover a linear nature of the transformations (under the gauge freedom), we define the shifted variables
\begin{equation}\label{eq:beta-shift}
\beta_E \;:=\; \alpha_E-\alpha_d,\qquad \beta \;:=\; \alpha-\alpha_d.
\end{equation}
Then \eqref{eq:affine-homothety} becomes linear again
\begin{equation}\label{eq:beta-linear}
\beta \;=\; e^{-w\lambda}\beta_E.
\end{equation}

\begin{remark}[Relation to compensator/dressing]
After the shift $\beta:=\alpha-\alpha_d$ (and similarly $\beta_E:=\alpha_E-\alpha_d$), the affine homothety about the fixed point  $\alpha_d$ becomes a purely multiplicative Weyl scaling, $\beta=e^{-w\lambda}\beta_E$. This is structurally the same mechanism as the standard compensator (dilaton) / Stueckelberg realization of local scale invariance, where one forms Weyl-covariant ``dressed'' variables by multiplying by an exponential of the scale field; see e.g.~\cite{Ghilencea2020WeylStueckelberg}. It also fits the general viewpoint of the dressing field method (DFM) for gauge-symmetry reduction, which has been applied in conformal/Weyl Cartan geometry (including the extraction of a dilaton from tractor data)~\cite{Attard2017DFM,Francois2019DilatonTractor}.
\end{remark}

The linearization may be achieved via adding the dimension, as we have done before \cite{Sabetghadam2025}. The resulting formulations are equivalent, but the set up of  \cite{Sabetghadam2025} is more clear for the physics applications (because it separates clearly the group of symmetry and the space of solutions). The present linearization by shifting stresses more the mathematical foundations, and therefore is more suitable for the PDE application. 

\noindent
Now, we can construct differential operators acting on such spaces.

\section{Twisted Operators} \label{sec:relation-witten}

We shall construct a Hodge-de Rham theory based on the above geometry. To this end, we define some twisted operators. Similar operators were used widely in the literature for the other purposes \cite{Witten1982,Novikov1981,Gover2005}.  

We fix a background (pseudo)-Riemannian metric $g_E$, in the Einstein frame, and twist the exterior calculus by the scalar field $\lambda$.
Define the twisted exterior derivative
\begin{equation}\label{eq:tilde-d}
\wtd \;:=\; e^{-w\lambda}\,\dd\,e^{w\lambda}
\;=\; \dd + w\,\dd\lambda\wedge .
\end{equation}
It satisfies the identity
\begin{equation}\label{eq:intertwine}
\wtd\!\left(e^{-w\lambda}\beta_E\right)=e^{-w\lambda}\dd\beta_E,
\end{equation}
and, since $\dd\phi=0$ with $\phi=\dd\lambda$, it is nilpotent:
\begin{equation}\label{eq:tilde-d-nilp}
\boxed{
\wtd^{\,2}=0.}
\end{equation}

Let $\delta$ and $\Delta=\dd\delta+\delta\dd$ denote the usual Hodge codifferential and Hodge Laplacian for the fixed metric $g_E$. We use the weighted pairing
\begin{equation}\label{eq:weighted-pairing}
\langle \eta,\xi\rangle_{\lambda,w} \;:=\; \int_M e^{2w\lambda}\,\eta\wedge *\xi,
\end{equation}
and define the corresponding twisted adjoint and Laplacian by conjugation:
\begin{equation}\label{eq:tilde-delta-tilde-Delta}
\wtdel \;:=\; e^{-w\lambda}\,\delta\,e^{w\lambda},
\qquad
\wtD \;:=\; \wtd\,\wtdel+\wtdel\,\wtd
\;=\; e^{-w\lambda}\,\Delta\,e^{w\lambda}.
\end{equation}
Conjugation immediately implies 
\begin{equation}\label{eq:nilp_delta}
\boxed{
\wtdel^{\,2}=0}
\end{equation}
and yields a self-adjoint elliptic operator $\tilde\Delta$ on $p$-forms under standard hypotheses (made precise later).

These three twisted operators are enough for our final goals. However, before developing the theory further, lets specify their relations with the other theories. 

\subsection{Relation to the other works}

Set
\[
f := w\lambda,\qquad S(\eta):=e^{f}\eta .
\]
Then our twisted differential
\[
\tilde d \;=\; e^{-f}\, d\, e^{f} \;=\; S^{-1} d S
\]
is precisely the \emph{Witten deformation} (exact twist) of the de~Rham differential \cite{Witten1982}.
More generally, one may twist by a closed $1$-form $\theta$ via $d_\theta:=d+\theta\wedge$, which yields the
Morse--Novikov/Lichnerowicz cohomology in the non-exact case \cite{Novikov1981,Farber2004}.
(We restrict throughout to the globally exact Weyl-integrable case $\phi=d\lambda$.)  Distinct, genuinely
\emph{conformally invariant} ``detour''/Branson--Gover-type de~Rham complexes also exist in conformal geometry,
but they are different objects from the present first-order twist \cite{BransonGover2003,Gover2005}.

\medskip
\noindent\textbf{Convention (adjoint vs.\ inner product).} There are two common ways to incorporate the Witten twist at the Hodge setting. In this paper we encode the twist in the Hilbert-space structure by using the weighted pairing \eqref{eq:weighted-pairing}, so that the formal adjoint and Laplacian satisfy the clean conjugation identities
\[
\tilde\delta = e^{-f}\,\delta\,e^{f}=S^{-1}\delta S,
\qquad
\tilde\Delta=\tilde d\,\tilde\delta+\tilde\delta\,\tilde d
= e^{-f}\,\Delta\,e^{f}=S^{-1}\Delta S .
\]
Many references instead keep the unweighted $L^2$ pairing and work with the corresponding adjoint
$d_f^{\,*}=e^{f}\delta e^{-f}$.

\section{Homothetic de Rham Complex} \label{sec:homothetic-derham}

\subsection{The homothetic (twisted) complex}

Fix $w\in\R$ and a smooth function $\lambda\in C^\infty(M)$ such that $\phi=\dd\lambda$.
Let $\wtd=e^{-w\lambda}\dd\, e^{w\lambda}$ be the twisted differential defined in \eqref{eq:tilde-d};
equivalently, $\wtd$ is the exact (Witten-type) twist of $\dd$ \cite{Witten1982}.
Since $\wtd^{\,2}=0$ (cf.\ \eqref{eq:tilde-d-nilp}), the sequence
\begin{equation}\label{eq:homothetic-complex}
0 \longrightarrow \Omega^0(M)\xrightarrow{\ \wtd\ }\Omega^1(M)\xrightarrow{\ \wtd\ }\cdots
\xrightarrow{\ \wtd\ }\Omega^d(M)\longrightarrow 0
\end{equation}
is a cochain complex, which we call the \emph{homothetic de~Rham complex}.

\begin{definition}[Homothetic de~Rham cohomology]\label{def:hom-cohom}
For $p=0,1,\dots,d$, define
\begin{equation}\label{eq:hom-cohom}
H^p_{\mathrm{hom}}(M;\lambda,w)
\;:=\;
\frac{\ker\big(\wtd:\Omega^p(M)\to\Omega^{p+1}(M)\big)}
     {\operatorname{im}\big(\wtd:\Omega^{p-1}(M)\to\Omega^{p}(M)\big)}.
\end{equation}
\end{definition}

\subsection{Chain isomorphism and relation to de Rham cohomology}
Define the multiplication operator
\begin{equation}\label{eq:S-operator}
S:\Omega^p(M)\to\Omega^p(M),\qquad S(\eta):=e^{w\lambda}\eta.
\end{equation}
Then $\wtd=S^{-1}\dd S$ and hence $S$ is a chain map between $\wtd$ and $\dd$:
\begin{equation}\label{eq:chain-map}
\boxed{
S\circ \wtd \;=\; \dd \circ S.}
\end{equation}
In particular, $S$ induces an isomorphism on cohomology whenever $\lambda$ is globally defined.

\begin{proposition}[Cohomological equivalence]\label{prop:cohom-equivalence}
Assume $\lambda\in C^\infty(M)$ is globally defined. Then multiplication by $e^{w\lambda}$ induces a canonical isomorphism
\begin{equation}\label{eq:Hiso}
H^p_{\mathrm{hom}}(M;\lambda,w)\;\cong\; H^p_{\dR}(M),
\qquad p=0,1,\dots,d.
\end{equation}
\end{proposition}

\begin{proof}
Let $\eta\in\Omega^p(M)$. Since $\wtd=S^{-1}\dd S$, we have $\wtd\eta=0$ if and only if $\dd(S\eta)=0$, and $\eta=\wtd\xi$ if and only if $S\eta=\dd(S\xi)$.
Thus $S$ maps $\ker\wtd$ onto $\ker\dd$ and $\operatorname{im}\wtd$ onto $\operatorname{im}\dd$, yielding \eqref{eq:Hiso}.
\end{proof}

\subsection{Homothetic closed and exact forms}
Recall the affine homothety $\alpha=e^{-w\lambda}(\alpha_E-\alpha_d)+\alpha_d$ and the shifted field $\beta=\alpha-\alpha_d=e^{-w\lambda}\beta_E$. The twisted differential $\wtd$ is designed so that $\wtd \beta$ transforms with the same weight as $\dd\beta_E$ in the Einstein gauge:
\begin{equation}\label{eq:intertwine-again}
\wtd \beta \;=\; \wtd\!\left(e^{-w\lambda}\beta_E\right)=e^{-w\lambda}\dd\beta_E. 
\end{equation}
Accordingly, in the homothetic gauge picture one naturally regards $\beta$ as \emph{homothetically closed} if $\wtd \beta=0$ and \emph{homothetically exact} if $\beta=\wtd \gamma$ for some $\gamma$.

\section{Homothetic Hodge Theory} \label{sec:homothetic-hodge}

\subsection{Weighted pairing and twisted adjoint}
Fix a background (pseudo)-Riemannian metric $g_E$ on $M$. Let $*$ denote the corresponding Hodge star, and let $\delta$ be the usual Hodge codifferential (the $L^2$-adjoint of $\dd$ with respect to $\int \eta\wedge *\xi$).Throughout this section, whenever we appeal to adjointness/integration by parts, we either assume $M$ is compact without boundary, or we work with compactly supported forms, so that the relevant boundary terms vanish and the pairings below are finite.
For $\lambda\in C^\infty(M)$ and weight $w\in\R$, define the weighted inner product
\begin{equation}\label{eq:weighted-inner-product}
\langle \eta,\xi\rangle_{\lambda,w}
\;:=\;
\int_M e^{2w\lambda}\,\eta\wedge *\xi,
\qquad \eta,\xi\in\Omega^p(M). \tag{\ref{eq:weighted-pairing} revisited}
\end{equation}
On noncompact $M$ this is understood on $\Omega^p_c(M)$, or on a weighted $L^2$-domain where the integral converges.

\begin{proposition}[Twisted adjoint by conjugation]\label{prop:twisted-adjoint}
With respect to \eqref{eq:weighted-pairing}, the formal adjoint of $\wtd=e^{-w\lambda}\dd e^{w\lambda}$ is
\begin{equation}\label{eq:tilde-delta}
\wtdel \;:=\; e^{-w\lambda}\,\delta\,e^{w\lambda}.
\end{equation}
And we have $\wtdel^{\,2}=0$. 
\end{proposition}

\begin{proof}
Let $\eta\in\Omega^{p}(M)$ and $\xi\in\Omega^{p+1}(M)$ with compact support (or assume hypotheses ensuring vanishing of boundary terms, e.g.\ $M$ compact without boundary). Set $\hat\eta:=e^{w\lambda}\eta$ and $\hat\xi:=e^{w\lambda}\xi$.
Then
\[
\langle \wtd\eta,\xi\rangle_{\lambda,w}
=
\int_M e^{2w\lambda}\,(e^{-w\lambda}\dd\hat\eta)\wedge *\xi
=
\int_M (\dd\hat\eta)\wedge *\hat\xi.
\]
By definition of $\delta$ as the adjoint of $\dd$ (for the unweighted pairing), this equals
\[
\int_M \hat\eta\wedge *(\delta\hat\xi)
=
\int_M e^{2w\lambda}\,\eta\wedge *\big(e^{-w\lambda}\delta(e^{w\lambda}\xi)\big)
=
\langle \eta,\wtdel\xi\rangle_{\lambda,w}.
\]
Finally, $\wtdel^{\,2}=0$ follows from $\delta^2=0$ and the conjugation \eqref{eq:tilde-delta}.
\end{proof}

\subsection{Homothetic Laplacian and conjugation}
Define the homothetic (Witten form) Laplacian by
\begin{equation}\label{eq:tilde-laplacian}
\wtD \;:=\; \wtd\,\wtdel+\wtdel\,\wtd.
\end{equation}
Let $\Delta:=\dd\delta+\delta\dd$ be the usual Hodge Laplacian for the fixed metric $g_E$.

\begin{proposition}[Conjugation identity]\label{prop:conjugation}
Let $S(\eta):=e^{w\lambda}\eta$.
Then
\begin{equation}\label{eq:conjugation}
\wtd = S^{-1}\dd S,\qquad
\wtdel = S^{-1}\delta S,\qquad
\wtD = S^{-1}\Delta S.
\end{equation}
\end{proposition}

\begin{proof}
The first two identities are the definitions of $\wtd$ and $\wtdel$.
For the Laplacian,
\[
\wtD
= (S^{-1}\dd S)(S^{-1}\delta S) + (S^{-1}\delta S)(S^{-1}\dd S)
= S^{-1}(\dd\delta+\delta\dd)S
= S^{-1}\Delta S.
\]
\end{proof}

\begin{remark}[Ellipticity and lower-order terms]\label{rem:elliptic}
Since $\wtD$ is conjugate to $\Delta$ by the pointwise multiplication $S=e^{w\lambda}$, it has the same principal symbol as $\Delta$ and hence is elliptic whenever $\Delta$ is elliptic (e.g.\ for Riemannian $g_E$). The dependence on $\lambda$ enters only through lower-order terms.
\end{remark}

\subsection{Homothetic harmonic forms and Hodge decomposition}
Define the space of homothetic harmonic $p$-forms by
\begin{equation}\label{eq:hom-harmonic}
\mathcal{H}^p_{\mathrm{hom}}(M;\lambda,w)
\;:=\;
\ker\big(\wtD:\Omega^p(M)\to\Omega^p(M)\big).
\end{equation}

\begin{proposition}[Transport of harmonic forms]\label{prop:harmonic-transport}
Assume $g_E$ is Riemannian. If $M$ is compact without boundary, or if one works in a class of forms for which $\wtD$ and $\Delta$ are defined and the conjugation \eqref{eq:conjugation} is valid (e.g.\ compact support or suitable decay/boundary conditions), then
\begin{equation}\label{eq:harmonic-transport}
\mathcal{H}^p_{\mathrm{hom}}(M;\lambda,w)
\;=\;
e^{-w\lambda}\,\mathcal{H}^p(M;g_E),
\end{equation}
where $\mathcal{H}^p(M;g_E)=\ker(\Delta:\Omega^p\to\Omega^p)$ denotes the usual space of harmonic $p$-forms for $g_E$.
\end{proposition}

\begin{proof}
By \eqref{eq:conjugation}, $\wtD\eta=0$ if and only if $\Delta(S\eta)=0$.
Thus $\eta\in\ker\wtD$ iff $S\eta=e^{w\lambda}\eta\in\ker\Delta$, i.e.\ iff $\eta=e^{-w\lambda}\omega$ with $\omega\in\mathcal{H}^p(M;g_E)$.
\end{proof}

\begin{theorem}[Homothetic Hodge decomposition]\label{thm:hom-hodge}
Assume $M$ is compact, oriented, and Riemannian (and without boundary).
Then for each $p$ there is an orthogonal decomposition with respect to $\langle\cdot,\cdot\rangle_{\lambda,w}$:
\begin{equation}\label{eq:hom-hodge-decomp}
\boxed{
\Omega^p(M) \;=\; \operatorname{im}\wtd \;\oplus\; \operatorname{im}\wtdel \;\oplus\;
\mathcal{H}^p_{\mathrm{hom}}(M;\lambda,w),}
\end{equation}
and every cohomology class in $H^p_{\mathrm{hom}}(M;\lambda,w)$ has a unique representative in
$\mathcal{H}^p_{\mathrm{hom}}(M;\lambda,w)$.
On noncompact manifolds such as $\R^d$, an analogous global orthogonal decomposition requires additional functional-analytic hypotheses (choice of Hilbert spaces, closures of ranges, and boundary/decay conditions), and is not used in the PDE applications below.
\end{theorem}

\begin{proof}[Proof sketch]
Since $\wtD=S^{-1}\Delta S$ is conjugate to the usual Hodge Laplacian, \eqref{eq:hom-hodge-decomp} follows by transporting the classical Hodge decomposition for $\Delta$ via the isomorphism $S$ and using that $\wtd=S^{-1}\dd S$ and $\wtdel=S^{-1}\delta S$.
A full proof is standard in weighted Hodge theory and is therefore omitted.
\end{proof}

\begin{remark}[Manifolds with boundary]\label{rem:boundary}
On manifolds with boundary, one may impose absolute/relative (or other) boundary conditions to obtain corresponding Hodge decompositions.
For the PDE applications below, we will mainly work on domains in $\R^d$ and discuss boundary/interface conditions directly at the level of the scalar homothetic Laplace equation.
In particular, on $\R^d$ we only use the scalar conjugation identity $\tilde\Delta u = e^{-w\lambda}\Delta(e^{w\lambda}u)$ and do not appeal to Theorem~\ref{thm:hom-hodge}.
\end{remark}

\section{A Homothetic Laplace equation for the Scalar Fields} \label{sec:scalar-pde}

From here on, we restrict ourselves on the Riemannian metrics, resulting in, the elliptic $\Delta$ and $\tilde{\Delta}$.
 
\subsection{0-forms: explicit form of the homothetic Laplacian}
In this section we specialize to $p=0$, i.e.\ scalar fields (functions).
Let $(M,g)$ be a Riemannian manifold with fixed metric $g$, and let $\Delta$ denote the (positive) Laplace--Beltrami operator on functions.\footnote{%
Throughout, one may equivalently use the geometric sign convention $\Delta=\delta\dd=-\mathrm{div}\,\grad$.} For $u\in C^\infty(M)$, the homothetic Laplacian is, by \eqref{eq:conjugation},
\begin{equation}\label{eq:scalar-tildeDelta-conj}
\wtD u \;=\; e^{-w\lambda}\,\Delta\!\big(e^{w\lambda}u\big).
\end{equation}
Expanding \eqref{eq:scalar-tildeDelta-conj} using the product rule yields the drift--potential form
\begin{equation}\label{eq:scalar-tildeDelta-expanded}
\wtD u \;=\; \Delta u + 2w\,\langle \grad\lambda,\grad u\rangle + \left(w(\Delta\lambda) + w^2\,|\grad\lambda|^2 \right)u,
\end{equation}
where $\langle\cdot,\cdot\rangle$ and $|\cdot|$ are taken with respect to $g$.

\begin{remark}\label{rem:div-form}
We should remark that 
\begin{equation}\label{eq:scalar-div-form}
\wtD u = 0
\qquad\Longleftrightarrow\qquad
\Delta\!\big(e^{w\lambda}u\big)=0.
\end{equation}
Thus the homothetic harmonicity of $u$ is equivalent to the classical harmonicity of the ``dressed'' field $e^{w\lambda}u$.
\end{remark}

\subsection{Homothetic Laplace equation for the homothetically transformed fields with the homothety center \texorpdfstring{$\alpha_d$}{alpha\_d}}
Let $\alpha_E$ be a scalar field in the Einstein gauge and let $\alpha_d$ be a fixed, scale-invariant $0$-form.
The homothetic transformation \eqref{eq:affine-homothety} reads
\begin{equation}\label{eq:scalar-affine}
\alpha \;=\; e^{-w\lambda}\alpha_E + (1-e^{-w\lambda})\alpha_d,
\end{equation}
and the shifted variable
\begin{equation}\label{eq:u-def}
u \;:=\; \alpha-\alpha_d
\end{equation}
transforms linearly: $u=e^{-w\lambda}(\alpha_E-\alpha_d)$.

\begin{definition}[Homothetic scalar Laplace equation]\label{def:hom-laplace-scalar}
We say that $\alpha$ satisfies the \emph{homothetic Laplace equation} (with center $\alpha_d$) if
\begin{equation}\label{eq:hom-laplace-u}
\wtD(\alpha-\alpha_d)=0, \qquad\text{i.e.} \qquad \wtD \alpha \;=\; \wtD\alpha_d.
\end{equation}
\end{definition}
\noindent
Expanding the above equation, and using \eqref{eq:scalar-tildeDelta-expanded} gives
\begin{equation}\label{eq:forced-expanded}
\Delta (\alpha-\alpha_d) + 2w\,\langle \grad\lambda,\grad (\alpha-\alpha_d)\rangle
+ \left( w(\Delta\lambda) + w^2\,|\grad\lambda|^2 \right)(\alpha-\alpha_d) = 0.
\end{equation}

And finally, if we consider the harmonic $\alpha_d$, in the classical sense, $\Delta\alpha_d=0$, then \eqref{eq:forced-expanded} reduces to
\begin{equation}\label{eq:forced-alpha-d-harmonic}
\boxed{
\Delta \alpha + 2w\,\langle \grad\lambda,\grad (\alpha-\alpha_d)\rangle
+ \left( w(\Delta\lambda) + w^2\,|\grad\lambda|^2 \right)(\alpha-\alpha_d) = 0.}
\end{equation}
We call this equation and its solutions the \emph{homothetic Laplace equation}, and the \emph{homothetic harmonic functions} respectively. This equation will be employed for the imposition of the boundary conditions on the classical Laplace equations in the next sections. 

\subsection{Brief remarks on well-posedness from the elliptic theory viewpoint}
On a bounded domain $\Omega\subset M$ with smooth boundary and Riemannian metric $g$, the operator
\begin{equation}\label{eq:Lu}
L_{\lambda,w}u \;:=\; \wtD u
=
\Delta u
+ 2w\,\langle \grad\lambda,\grad u\rangle
+ w(\Delta\lambda)\,u
+ w^2\,|\grad\lambda|^2\,u
\end{equation}
is a second-order elliptic operator with the same principal symbol as $\Delta$.
In particular, standard existence and uniqueness results apply under standard boundary conditions (Dirichlet, Neumann, or mixed) provided the coefficients determined by $\lambda$ are sufficiently regular.
In later sections we focus on a complementary viewpoint: choosing $\lambda$ to localize near an interface so that the single equation $L_{\lambda,w}u=0$ effectively enforces boundary or matching data for $u=\alpha-\alpha_d$.

\section{PDE applications}\label{PDE-Appl}

The scalar homothetic Laplace equation \eqref{eq:forced-alpha-d-harmonic} can be used to \emph{impose} boundary/interface conditions for classical elliptic problems \cite{Evans2010,GilbargTrudinger2001}. The idea is to choose $\lambda$ so that the lower-order coefficients in the homothetic operator concentrate near a hypersurface $S$, to penalize the mismatch of prescribed data encapsulated in $\phi_d$. Thus $\lambda$ is treated as a given fixed parameter, not a dynamical one. This viewpoint is particularly appropriate for PDE applications, and geometrically justifies the volume-penalization/diffuse-interface approaches are widely in use in the computational physics \cite{AngotBruneauFabrie1999,Mittal2008,Nguyen2014}. We use the normalizing choice $w=1$, which yields particularly transparent formulas. We restrict to the Euclidean domains $\R^d$, on which the inner products simplify to the dot products. Moreover, to emphasize that we are treating 0-forms, we will use $\phi$ instead of $\alpha$. Therefore, Eq. \eqref{eq:forced-alpha-d-harmonic} reads 

\begin{equation}\label{eq:forced-alpha-d-harmonic-Euc}
\Delta \phi + 2\grad\lambda\cdot\grad (\phi-\phi_d)+ \left( (\Delta\lambda) + |\grad\lambda|^2 \right)(\phi-\phi_d) = 0.
\end{equation}

We always assume $\phi_d$ is harmonic in the classical sense $\Delta \phi_d=0$. For instance, $\phi_d$ may be constructed in a tubular neighborhood of $S$ as the harmonic extension of the prescribed boundary data. Then, by its smooth extension outside this neighborhood, the penalization terms below are localized near $S$ (cf.\ standard elliptic regularity and extension results \cite{GilbargTrudinger2001,Evans2010}).

\subsection{Cauchy, Dirichlet and Neumann problems}\label{subsec:7.1}

Let $S\subset \R^d$ be a closed, connected, orientable, embedded $C^\infty$ hypersurface where $d\in\{1,2,3\}$. By the Jordan--Brouwer separation theorem, $\R^d\setminus S$ has precisely two connected components, one bounded and one unbounded \cite{Munkres2000}. We denote them by $\Omega_i$ (interior) and $\Omega_o$ (exterior), and we have ${\Omega}_i\cap {\Omega}_o=S$. Let $\nu$ be the unit normal on $S$ pointing from $\Omega_i$ into $\Omega_o$, and for a sufficiently regular function $u$ define $\partial_\nu u := \nabla u\cdot \nu$ on $S$. Moreover, let $T_{\varepsilon_n}(S)$ denote the tubular neighborhood of $S$ of radius $\varepsilon_n$, where $\lim_{n\to\infty} \varepsilon_n=0$. Let $\xi(x):=\mathrm{dist}(x,S)$ denote the (unsigned) distance function, which is smooth in each of the two components $\Omega_i$ and $\Omega_o$ away from $S$. Using these definitions, we set
\begin{equation}\label{eq:lambda-n}
\lambda_n := \ln f_n(\xi(x)),
\end{equation}
in which $f_n$ belongs to a family $(f_n)_{n\ge 1}$ of smooth cutoff functions (see \ref{lambda_bcs} for the details). Note that such a behavior results in
\begin{equation}\label{eq:pen_11}
\nabla\lambda_n=\frac{\nabla f_n}{f_n}\qquad\text{and}\qquad \Delta\lambda_n+|\nabla\lambda_n|^2=\frac{\Delta f_n}{f_n}.
\end{equation}
In particular, since $f_n\equiv 1$ on $\R^d\setminus T_{\varepsilon_n}(S)$, one has $\nabla\lambda_n\equiv 0$ and $\Delta\lambda_n+|\nabla\lambda_n|^2\equiv 0$ there, so all lower--order homothetic terms vanish identically away from $T_{\epsilon_n}(S)$.

\medskip
\noindent\textbf{Cauchy data on $S$.} Fix boundary data
\[
g := \phi_S \quad\text{and}\quad h := \partial_\nu \phi_S
\]
on $S$, the problem 
\begin{equation}\label{eq:Cauchy-strong}
\phi(x):=\begin{cases}
\Delta \phi = 0 \ \ \text{in }\R^d \\
\phi|_{S}=g, \\
\partial_\nu\phi|_{S}=h,
\end{cases}
\end{equation}
under a suitable condition at infinity in the exterior domain (e.g.\ decay, boundedness, or a prescribed far--field behavior), defines a \emph{Cauchy problem} for the Laplace equation on $\R^d$. 

It is well known that \eqref{eq:Cauchy-strong} is overdetermined. Generically, a pair $(g,h)$ does \emph{not} arise as the trace and normal trace of a global harmonic function on $\R^d$, and the corresponding continuation problem is ill-posed in the Hadamard sense (see, e.g., \cite{Hadamard1923,LionsMagenes1972}).

\medskip
\noindent \textbf{A penalized homothetic formulation.}
Given a fixed function $\phi_d$ which is harmonic in $T_{\epsilon_n}(S)$ and matches the desired boundary data on $S$
\[
\phi_d|_S=g,\qquad \partial_\nu\phi_d|_S=h,
\]
By choosing $\lambda_n$ from \eqref{eq:lambda-n} resulting in \eqref{eq:pen_11} coefficients, Eq. \eqref{eq:forced-alpha-d-harmonic-Euc} becomes the \emph{penalized Cauchy problem}
\begin{equation}\label{eq:penalized-Cauchy}
\boxed{\Delta \phi_n
+2\,\frac{\nabla f_n}{f_n}\cdot\bigl(\nabla\phi_n-\nabla\phi_d\bigr)
+\frac{\Delta f_n}{f_n}\,(\phi_n-\phi_d)=0}
\qquad \text{in }\R^d\setminus S.
\end{equation}
For each fixed $n$, \eqref{eq:penalized-Cauchy} is a linear second--order elliptic equation with smooth coefficients on $\Omega_i$ and on $\Omega_o$. As $n\to\infty$, the lower--order terms concentrate near $S$ and act as a thin penalty layer which attempts to enforce \emph{simultaneously} the Dirichlet and Neumann data encoded in $\phi_d$.

\medskip
\noindent\textbf{Strong vs.\ weak solutions.} Since, in the limit, the solutions may fail to be $C^1$ across $S$, it is convenient to work with a piecewise Sobolev space  \cite{AdamsFournier2003,LionsMagenes1972}. Define the space
\[
H^1(\Omega_i\sqcup \Omega_o)
:=\bigl\{u:\ u|_{\Omega_i}\in H^1(\Omega_i),\ u|_{\Omega_o}\in H^1(\Omega_o)\bigr\},
\qquad
\|u\|^2 := \|u\|^2_{H^1(\Omega_i)}+\|u\|^2_{H^1(\Omega_o)}.
\]
For $u\in H^1(\Omega_i\sqcup\Omega_o)$, denote by $\gamma_i u$ and $\gamma_o u$ the (Sobolev) traces on $S$ from the interior and exterior side, respectively. When the traces exist, we define the \emph{jump} by
\[
[u]_S := \gamma_o u-\gamma_i u .
\]
Similarly, for sufficiently regular $u$ (e.g.\ piecewise $H^2$), one may define the normal derivative traces $\partial_\nu u|_{\Omega_i}$ and $\partial_\nu u|_{\Omega_o}$ and their jump $[\partial_\nu u]_S := (\partial_\nu u|_{\Omega_o})-(\partial_\nu u|_{\Omega_i})$.

\noindent
With these definitions in place, we analyze Eq. \eqref{eq:penalized-Cauchy}.

\medskip
\noindent\textbf{Consistent and inconsistent Cauchy data.} For Eq. \eqref{eq:penalized-Cauchy} two different situations are distinguishable.

\begin{itemize}
\item[\it{(A)}] \emph{Consistent Cauchy data (a global harmonic function exists).} \\
Assume that there exists a harmonic function $\phi\in C^\infty(\R^d)$ satisfying the Cauchy data,
that is, \eqref{eq:Cauchy-strong} holds (together with the chosen behavior at infinity in the
exterior domain). In this situation, the penalized Cauchy formulation \eqref{eq:penalized-Cauchy} is consistent with a global classical harmonic function. It is then natural (and can be proved under standard uniform elliptic estimates) that the solutions $\phi_n$ converge, as $n\to\infty$, to that global harmonic function in the \emph{strong} sense (e.g., in $H^1_{\mathrm{loc}}(\R^d)$, and in $C^\infty_{\mathrm{loc}}(\R^d\setminus S)$ by elliptic regularity). In particular, the restrictions
\[
\phi_i := \phi|_{\Omega_i},\qquad \phi_o := \phi|_{\Omega_o}
\]
are harmonic and satisfy the full matching conditions
\[
[\phi]_S=0,\qquad [\partial_\nu\phi]_S=0,
\]
so that $\phi_i$ and $\phi_o$ are harmonic continuations of one another across $S$.

\begin{remark}[Internal surface] In this case, the hypersurface $S$ does not act as a boundary for the limiting harmonic field: the interior and exterior solutions glue to a single $\phi\in C^\infty(\R^d)$ with $[\phi]_S=[\partial_\nu\phi]_S=0$. In this case, $S$ is an \emph{internal surface} (or \emph{removable interface}).
\end{remark}

\item[\it(B)] \emph{Inconsistent Cauchy data (no global harmonic function exists).}\\
Assume now that the Cauchy data are inconsistent, that is, there is no global harmonic function on
$\R^d$ satisfying both $\phi|_S=g$ and $\partial_\nu\phi|_S=h$. In other words, one cannot enforce \emph{simultaneously} the two constraints encoded in $\phi_d$ and $\nabla\phi_d$ on $S$.

\begin{definition}[Homothetic boundary] In this case, we call $S$ a \emph{homothetic boundary}. Equivalently, there is no global harmonic $\phi$ on $\R^d$ satisfying both $\phi|_S=g$ and $\partial_\nu\phi|_S=h$, so the limiting homothetic formulation necessarily yields a problem in which at least one of the jumps $[\phi]_S$ or $[\partial_\nu\phi]_S$ is nonzero.
\end{definition}

When there is not a consistent Cauchy data in hand, it is usually natural to think of imposing only one of the two boundary conditions. It leads to the following two reduced penalized problems.
\end{itemize}

\begin{itemize}
\item \textbf{Penalized Dirichlet problem.}\\
Solve
\begin{equation}
\label{eq:penalized-Dirichlet}
\Delta \phi_n + \frac{\Delta f_n}{f_n}\,(\phi_n-\phi_d)=0\qquad \text{in }\R^d\setminus S.
\end{equation}
Formally, the penalty term forces $\gamma_i\phi_n\approx \gamma_o\phi_n\approx g$ on $S$ as
$n\to\infty$, but it does \emph{not} enforce matching of normal derivatives. Hence, in the limit one generically obtains a piecewise harmonic function $\phi\in H^1(\Omega_i\sqcup\Omega_o)$ satisfying
\[
\gamma_i\phi=\gamma_o\phi=g \quad\text{on }S,
\]
while the normal derivative may jump:
\[
[\partial_\nu\phi]_S\neq 0 \quad \text{in general}.
\]
Consequently there is, in general, \emph{no} global strong (classical) harmonic solution on $\R^d$; rather, the limit object is a weak solution in the space $H^1(\Omega_i\sqcup\Omega_o)$.

\item \textbf{Penalized Neumann problem.} \\
Here we assume $h$ satisfies the conventional consistency condition for a homogeneous Neumann problem $\int_S hdS=0$. Solve 
\begin{equation}
\label{eq:penalized-Neumann}
\Delta \phi_n + 2\,\frac{\nabla f_n}{f_n}\cdot\bigl(\nabla\phi_n-\nabla\phi_d\bigr)=0
\qquad \text{in }\R^d\setminus S.
\end{equation}
Here the penalization enforces $\partial_\nu\phi|_S=h$ in the limit (in the natural weak trace sense), but it does \emph{not} enforce matching of values across $S$. Thus the limiting object is again piecewise harmonic, satisfies the Neumann condition on $S$ (on each side, with the appropriate sign convention), and may exhibit a jump in the trace:
\[
[\phi]_S\neq 0 \quad \text{in general}.
\]
As above, this precludes a global strong solution on $\R^d$ in general, but admits a weak solution; cf.\ the standard variational theory for Neumann problems \cite{LionsMagenes1972,Evans2010}.
\end{itemize}

\noindent
The above issues can be summarized as follows.

\begin{theorem}\label{thm:penalization-limit}
Given the above definitions, assume $\phi_n$ is a family of solutions to either \eqref{eq:penalized-Dirichlet} or \eqref{eq:penalized-Neumann}, subject to a fixed behavior at infinity in the exterior domain. Assume further that $\phi_n$ is
uniformly bounded in $H^1_{\mathrm{loc}}(\Omega_i)$ and $H^1_{\mathrm{loc}}(\Omega_o)$. Then, after extraction of a subsequence, $\phi_n$ converges weakly in $H^1_{\mathrm{loc}}(\Omega_i)\times H^1_{\mathrm{loc}}(\Omega_o)$ to a pair $(\phi_i,\phi_o)$ such that:
\begin{enumerate}
\item $\Delta \phi_i = 0$ in $\Omega_i$ and $\Delta\phi_o=0$ in $\Omega_o$ in the weak sense (and hence $\phi_i,\phi_o$ are smooth by elliptic regularity).
\item In the Dirichlet case \eqref{eq:penalized-Dirichlet}, the traces satisfy $\gamma_i\phi_i=\gamma_o\phi_o=g$ on $S$, i.e.\ $(\phi_i,\phi_o)$ solve the classical Laplace--Dirichlet problems on $\Omega_i$ and $\Omega_o$ with boundary value $g$.
\item In the Neumann case \eqref{eq:penalized-Neumann}, the (weak) normal traces satisfy the prescribed Neumann condition determined by $h$ on $S$ (with the standard sign conventions on $\Omega_i$ and $\Omega_o$), i.e.\ $(\phi_i,\phi_o)$ solve the classical Laplace--Neumann problems on $\Omega_i$ and $\Omega_o$ with boundary flux $h_d$ (with the usual additive-constant ambiguity, fixed by the boundary condition at infinity if needed).
\end{enumerate}
In general, the global glued function
\[
\phi(x):=\begin{cases}
\phi_i(x), & x\in \Omega_i,\\
\phi_o(x), & x\in \Omega_o,
\end{cases}
\]
is a weak solution on $\R^d$ and may exhibit a jump in $\phi$ or in $\partial_\nu\phi$
across $S$, depending on whether Dirichlet or Neumann data were imposed.
\end{theorem}
The above issues can be more developed if we put them inside the framework of classical potential theory.

\subsection{Relation to classical potential theory} 

The above discussions are closely related to classical potential theory, particularly the behavior of single-layer and double-layer potentials \cite{Kress2014, Ladyzhenskaya1985}; see also \cite{McLean2000}. Specifically, a \emph{single-layer potential} yields a continuous potential across the interface but typically introduces a discontinuity in its normal derivative. In contrast, a \emph{double-layer potential} maintains continuity of the normal derivative but allows for a jump in the potential itself \cite{ColtonKress1998}. The behavior of equation \eqref{eq:penalized-Cauchy} in the limit $n \to \infty$ is thus analogous to a \emph{combined layer potential}, where both Dirichlet and Neumann data are enforced. Therefore, the above discussions can be rephrased in the following theorems.

\begin{proposition}[Distributional Laplacian across a hypersurface]
\label{prop:distributional-Laplace-jumps}
Let $S\subset\R^d$ be a $C^\infty$ embedded hypersurface and write $\R^d=\Omega_i\sqcup S\sqcup\Omega_o$. Fix a unit normal field $\nu$ on $S$ pointing from $\Omega_i$ to $\Omega_o$ and write $\partial_\nu:=\nabla(\,\cdot\,)\cdot \nu$. Let $\phi_i\in H^2_{\mathrm{loc}}(\Omega_i)$ and $\phi_o\in H^2_{\mathrm{loc}}(\Omega_o)$ satisfy
\[
\Delta \phi_i=0\ \text{in }\Omega_i,\qquad \Delta\phi_o=0\ \text{in }\Omega_o,
\]
and define the glued function $\phi$ on $\R^d$ by
\[
\phi(x):=\begin{cases}
\phi_i(x),&x\in\Omega_i,\\
\phi_o(x),&x\in\Omega_o.
\end{cases}
\]
Assume the traces $\gamma_i\phi,\gamma_o\phi$ and normal traces exist so that the jumps
\[
[\phi]_S:=\gamma_o\phi-\gamma_i\phi,\qquad
[\partial_\nu\phi]_S:=\bigl(\partial_\nu\phi_o|_S\bigr)-\bigl(\partial_\nu\phi_i|_S\bigr)
\]
are well-defined. Then, for every $\psi\in C_c^\infty(\R^d)$ one has
\begin{equation}
\label{eq:distributional-Laplace-jumps}
\langle \Delta \phi,\psi\rangle= \int_S [\partial_\nu\phi]_S\,\psi\,d\sigma \;-\; \int_S [\phi]_S\,\partial_\nu\psi\,d\sigma .
\end{equation}
Equivalently, in $\mathcal D'(\R^d)$,
\[
\Delta\phi = [\partial_\nu\phi]_S\,\delta_S \;+\; [\phi]_S\,\partial_\nu\delta_S,
\qquad\text{where}\qquad \langle \partial_\nu\delta_S,\psi\rangle:=-\int_S \partial_\nu\psi\,d\sigma .
\]
In particular, $\phi$ is a global harmonic function on $\R^d$ if and only if $[\phi]_S=0$ and $[\partial_\nu\phi]_S=0$ (i.e.\ $S$ is a removable interface for $\phi$).
\end{proposition}

\begin{corollary}[Potential theory setting] \label{cor:penalization-produces-layer}
In the setting of Theorem~\ref{thm:penalization-limit}, assume additionally that the limiting pair
$(\phi_i,\phi_o)$ admits traces and normal traces on $S$ so that the jumps $[\phi]_S$ and $[\partial_\nu\phi]_S$ are well-defined, and let $\phi$ be the glued function. Then $\Delta\phi$ is a distribution supported on $S$ and is given by
\[
\Delta\phi =[\partial_\nu\phi]_S\,\delta_S \;+\; [\phi]_S\,\partial_\nu\delta_S \qquad\text{in }\mathcal D'(\R^d).
\]
Moreover:
\begin{enumerate}
\item In the Dirichlet penalization \eqref{eq:penalized-Dirichlet}, one has $\gamma_i\phi_i=\gamma_o\phi_o=g$ on $S$, hence $[\phi]_S=0$ and
\[
\Delta\phi = [\partial_\nu\phi]_S\,\delta_S ,
\]
i.e.\ $S$ carries a \emph{single-layer} source unless $[\partial_\nu\phi]_S=0$.

\item In the Neumann penalization \eqref{eq:penalized-Neumann}, the enforced flux condition yields
the corresponding normal-trace constraint on $S$ (with the sign convention of
Theorem~\ref{thm:penalization-limit}); in the flux-matching case $[\partial_\nu\phi]_S=0$ and
\[
\Delta\phi = [\phi]_S\,\partial_\nu\delta_S ,
\]
i.e.\ $S$ carries a \emph{double-layer} source unless $[\phi]_S=0$.
\end{enumerate}
Consequently, whenever at least one of the jumps is nonzero, $S$ is detected by the classical
Laplacian as a genuine geometric boundary/interface for the limit $n\to\infty$.
\end{corollary}

In the next section we employ the developed theory to address the longstanding problem of non-singular point source in the potential theory.
\subsection{A non-singular model for a point source}\label{subsec:7.2}

Homothetic harmonic functions are more flexible than the classical harmonic functions due to the presence of $ \lambda $ and $ \phi _ d$, which allow local deformations. For example, from Liouville’s theorem, one can infer that any harmonic function on $\R^d$ is either constant or not bounded. But this is not true for the homothetic harmonic functions. Here, we utilise this flexibility to construct a non-singular model for a point source. In fact, there is more than one type of such regularization. We have already provided a non-singular model for a point source as a solution of Eq. \eqref{eq:Cauchy-strong}, suitable for an elementary particle (an electron in particular) \cite{Sabetghadam2025}. Here, another non-singular model is presented based on Eq. \eqref{eq:penalized-Dirichlet}, which models a hollow-sphere point source (c.f. classical electrostatics and potential theory \cite{Jackson1998,Kellogg1929}.) 

The potential of a point charge located at the origin of \( \mathbb{R}^3 \) takes the form
\[
\phi(r) = \frac{C}{r},
\]
where \( C \) is a constant and \( r = \| \mathbf{x} \| \). While this solution satisfies the Laplace equation away from the origin, it becomes singular at \( r = 0 \), and the corresponding energy,
\[
E = \frac{1}{2} \int_{\mathbb{R}^3} |\nabla \phi|^2 \, dV,
\]
diverges due to the singularity at the origin (indeed $|\nabla\phi(r)|\sim r^{-2}$ and $\int_0^1 r^{-2}\,dr=+\infty$). This infinite self-energy has been a foundational issue in classical field theory, prompting numerous attempts at regularization \cite{Jackson1998,BornInfeld1934}. Now, in the framework of homothetic harmonic functions one can avoid such singularities by enforcing the boundary data on a spherical interface rather than at a point.

We consider a spherical interface $\partial B(R) \subset \mathbb{R}^3$, and solve the Dirichlet homothetic Laplace equation \eqref{eq:penalized-Dirichlet} adapted for this problem
\begin{equation}\label{eq:dir_Lap_sphere}
\nabla^2 \phi_n + \frac{\Delta f_{R,n}}{f_{R,n}}\,(\phi_n-\phi_d)=0,
\end{equation}
where $f_{R,n}$ is the cutoff function associated with the surface $\partial B(R)$ constructed from the distance-to-$\partial B(R)$ as explained in \ref{lambda_bcs}, that is, $f_{R,n}(x)=\mathrm{dist}(x,\partial B(R))$ in a layer of thickness $\eta\varepsilon_n$ in the neighborhood of $S$, and $f_{R,n}\equiv 1$ outside a layer of thickness $\varepsilon_n$; and $\phi_d$ is the prescribed potential
\begin{equation}
\phi_d=\frac{C}{R},
\end{equation}
enforcing a Coulomb-like condition at $r=R$, where $r=\|x\|$.

\begin{remark}
It should be emphasized that in contrast to the classical point source models, the penalization layer here is supported near the surface $\partial B(R)$, not at the origin. For every fixed $R>0$ the interface is a smooth hypersurface $\partial B(R)\cong S^2$, not a single point. From the viewpoint of potential theory, this distinction is reflected in capacity: the Newtonian capacity of a sphere of radius $R$ is proportional to $R$ (in fact, $\mathrm{cap}(\partial B(R))=4\pi R$ in $\R^3$), whereas a point has zero Newtonian capacity; see, e.g., \cite{Kellogg1929,Landkof1972}.
\end{remark}

Due to the highly localized nature of the penalty term, the solution $\phi_n$ adjusts sharply near
the interface to match the boundary value, and decays as $\phi_n(r)\sim C/r$ for large $r$, consistent with the behavior of the classical potential in the far field  \cite{Jackson1998,Kellogg1929}.

The structure of the equation ensures a natural decoupling of the interior and exterior domains. In the limit $n\to\infty$, the penalization enforces the boundary condition only on $\partial B(R)$, and the interior and exterior solutions become independent (see the discussion in \S \ref{subsec:7.1}). Thus, the behavior of the solution in $\mathbb{R}^3\setminus B(R)$ is determined entirely by the enforced Dirichlet data and is unaffected by the field inside the ball. In particular, the exterior solution is the unique harmonic function on $\mathbb{R}^3\setminus \overline{B(R)}$ that decays at infinity and satisfies $\phi|_{r=R}=C/R$ \cite{Evans2010,GilbargTrudinger2001}, hence
\begin{equation}
\phi_o(r)=\frac{C}{r},\qquad r\ge R.
\end{equation}
On the other hand, the interior solution on $B(R)$ is determined solely by the same boundary condition on $\partial B(R)$; the unique regular harmonic solution in the ball satisfying $\phi|_{r=R}=C/R$ is the constant potential (by the maximum principle) \cite{Evans2010,GilbargTrudinger2001}
\begin{equation}
\phi_i(r)=\frac{C}{R},\qquad r\le R.
\end{equation}
In particular, $\phi_i$ and $\phi_o$ model a hollow spherical source of radius $R$ with the surface potential $C/R$ placed at the origin of $\R^3$, and they arise as the limiting interior/exterior fields associated with the single penalized equation \eqref{eq:dir_Lap_sphere}. However, this solution allows a jump in the normal derivative across $\partial B(R)$, in accordance with the discussion in \S \ref{subsec:7.1}.

One may now observe the behavior of the model for $R\downarrow 0$. For each fixed $R>0$, the field energy is finite, since $\nabla\phi_i\equiv 0$ in $B(R)$ and $|\nabla\phi_o(r)|=|C|\,r^{-2}$ in the exterior, hence
\begin{equation}
E =\frac{1}{2}\int_{\mathbb{R}^3}|\nabla \phi|^2\,dV
=\frac{1}{2}\int_{\R^3\setminus B(R)} \frac{C^2}{r^4}\,dV
=2\pi\,\frac{C^2}{R}<\infty\qquad (R>0).
\end{equation}
In particular, the model removes the singularity at $r=0$ for every $R>0$. Therefore, by subtraction of the single point $r=0$ from $\R^3$, the energy remains non-infinite on $\R^3 \setminus \{0\}$ for all $R$. 

\section{Conclusions}\label{sec:Conclusions}

In this paper, we have proposed an extension to classical Weyl-integrable geometry by generalizing standard linear gauge transformations into affine homothetic transformations centered around a distinguished, harmonic, scale-invariant form $\alpha_d$. By defining shifted variables to recover the linearity, we constructed a twisted Hodge--de Rham theory. We demonstrated that our twisted differential operators inherently correspond to Witten-type exact deformations, enabling us to establish a homothetic Hodge decomposition theorem for compact, oriented Riemannian manifolds.

Transitioning from the abstract geometric framework to practical PDE applications, we specialized this theory for $0$-forms to derive the scalar homothetic Laplace equation. We showed that by treating the twist parameter $\lambda$ as a fixed, non-dynamical field determined by a regularized distribution tuned to concentrate near a hypersurface, this equation acts as a powerful diffuse-interface tool for penalizing mismatches in prescribed boundary data. Specifically, setting $\lambda$ via a regularized surface delta distribution allows the single homothetic equation to enforce Cauchy, Dirichlet, or Neumann boundary conditions. In cases of incompatible Cauchy data, we proved that the penalized equations naturally relax to yield weak solutions, permitting jumps in either the potential or its normal derivative in a manner reminiscent of classical single and double layer potentials.

Furthermore, we highlighted the physical relevance of this geometric flexibility by constructing a non-singular model for a classical point source. By modeling a hollow-sphere regularization that enforces a surface potential on an interface $\partial B(R)$, our approach entirely avoids the infinite self-energy at the origin. The penalization cleanly decouples the interior and exterior domains in the limit, ensuring a constant finite interior potential while preserving the physically observable decay in the far field. 

Finally, our analytical approach to determining the highly localized profile of $\lambda$—utilizing the well-established theory of regular singular points for second-order ODEs—provides a systematic (if not quite general) approach for boundary data imposition. This methodology not only resolves foundational singularities in classical field theory but also demonstrates significant potential for future algorithmic implementations in computational physics.

\appendix
\section{Tuning $\lambda$ near the boundary}\label{lambda_bcs}

In $\S$ \ref{PDE-Appl}, $\lambda(x)$ is a given fixed function, that is, it does not have a separate dynamics. In this appendix, we find this fixed distribution in the neighborhood of the boundary, which allows us to apply our desired boundary conditions. The result has the potential for applications in computational physics. Our goal here is to provide a practical approach for obtaining sufficiently regular boundaries/interfaces, rather than a complete classification of all admissible near-boundary behavior. We employ the well-established theory of regular singular points of second-order ordinary differential equations (see, e.g., \cite{CoddingtonLevinson1955,Simmons2016}). 

Consider the one-dimensional version of the homothetic Laplace equation \eqref{eq:forced-alpha-d-harmonic-Euc}, and set $u:=\phi-\phi_d$ to get:
\begin{equation}\label{eq:A_u}
u''+2\lambda' u' + (\lambda''+(\lambda')^2)\,u=0.
\end{equation}
where $(\cdot^\prime)$ denotes differentiation with respect to $x\in (0,\infty)$, and $x=0$ is the point that we want to impose our boundary conditions. Moreover, we assumed $\phi_d$ is harmonic $\phi_d''=0$ in the neighborhood $x=0^+$.

The coefficients $\lambda'$ and $\lambda''+(\lambda')^2$ should become large as $x\to 0^+$, and vanish away from it. A standard sufficient condition for $x=0$ to be a \emph{regular singular point} of \eqref{eq:A_u} is the Fuchsian scaling
\begin{equation}\label{eq:A_fuchs}
\lambda'(x)\sim \frac{a}{x},
\qquad
\lambda''(x)+(\lambda'(x))^2 \sim \frac{a^2-a}{x^2}
\qquad (x\to 0^+),
\end{equation}
for some real parameter $a\in\R$. We interpret $a$ as a calibration parameter controlling the sharpness at the boundary. Moreover, we can vanish the lower-order terms, away from $x=0$, as follows.

Let $(\varepsilon_n)_{n\ge 1}$ be a sequence with $\varepsilon_n \downarrow 0$ defining the collar neighborhood of the origin $x=0$, and fix $\eta\in(0,1)$. Choose a family of functions $f_n:(0,\infty)\to(0,1]$ such that
\begin{equation}\label{eq:A_fn}
f_n(x)=\begin{cases}
x,& \ \ \text{for } 0<x\le \eta\,\varepsilon_n, \\
1,& \ \ \text{for } x\ge \varepsilon_n,
\end{cases}  \qquad
f_n\in C^\infty(~]0,\infty[~).
\end{equation}
Such $f_n$ are obtained by composing $x$ with a standard $C^\infty$ cutoff profile that interpolates smoothly between $x$ and $1$ on $[\eta\varepsilon_n,\varepsilon_n]$ (see e.g.\cite{Hormander1983}). We then define
\begin{equation}\label{eq:A_lambda_def}
\boxed{\ \lambda_n(x) := a\,\ln\!\bigl(f_n(x)\bigr).}
\end{equation}
In the inner part $0<x\le \eta\varepsilon_n$ where $f_n(x)=x$, we obtain the regular singular behavior, obeying the Euler--Cauchy equation
\begin{equation}\label{eq:A_euler}
u''+\frac{2a}{x}\,u' + \frac{a^2-a}{x^2}\,u=0
\qquad (0<x\le \eta\varepsilon_n),
\end{equation}
while in the outer part $x \geq \varepsilon_n$, the coefficients vanish identically and the Laplace equation $u''=0$ retrieves.

Trying a Frobenius/Euler ansatz $u(x)=C x^m$ gives the indicial equation
\[
m(m-1)+2a\,m+(a^2-a)=0,
\]
whose roots are
\[
m_1=-a,
\qquad
m_2=1-a.
\]
Hence the local behavior near $x=0^+$ is
\begin{equation}\label{eq:A_local_sol}
u(x)\sim C_1 x^{-a}+C_2 x^{\,1-a}.
\end{equation}
The parameter $a$ controls the vanishing/blow-up rates of these two branches. The branch that is selected in a given PDE application is determined by the criterion one imposes (e.g.\ boundedness, finite energy, or a weak trace requirement). If one selects a branch $u\sim x^{m}$ with exponent $m>0$, then $u(0)=0$ (suitable for the Dirichlet BCs). More generally, $u^{(k)}(0)=0$ holds provided $m>k$. For example, for the branch $u\sim x^{-a}$ one has
\[
u(0)=0 \ \text{if } a<0,\qquad
u'(0)=0 \ \text{if } a<-1,\qquad
u''(0)=0 \ \text{if } a<-2,
\]
while for the branch $u\sim x^{1-a}$ one has
\[
u(0)=0 \ \text{if } a<1,\qquad
u'(0)=0 \ \text{if } a<0,\qquad
u''(0)=0 \ \text{if } a<-1.
\]
Thus $a$ should be viewed as a tunable penalty-strength parameter, by which one can determine how many derivatives of the selected local branch vanish at the singular point.\\

\medskip
\noindent\textbf{Multi-dimensional interfaces.}
For a smooth embedded hypersurface $S\subset\R^d, d>1$, one works in a tubular neighborhood $T_{\varepsilon_n}(S)$ with distance coordinate $\xi(x)=\mathrm{dist}(x,S)$, and sets
\[
f_n(x)=F_n(\xi(x)), \qquad \lambda_n(x)=a\,\ln\!\bigl(f_n(x)\bigr),
\]
where $F_n$ is the one-dimensional profile from \eqref{eq:A_fn}. In local normal coordinates, the leading singular behavior of the homothetic coefficients is governed by the same $\xi^{-1}$ and $\xi^{-2}$ scalings as in the one-dimensional Euler--Cauchy model, while curvature and tangential contributions enter as lower-order terms. Consequently, the above one-dimensional analysis provides (at least in principle) a local procedure for tuning $a$. However, we do not assert a universal dimension-independent threshold for $a$ in full generality. Moreover, note that $\xi$ is smooth on each side in $\Omega_i$ and $\Omega_o$ (in $\S$ \ref{PDE-Appl} setting) within the tubular neighborhood of $S$.


\end{document}